\documentclass[11pt]{amsart}
\usepackage{amssymb,color}
\numberwithin{equation}{section}
\usepackage{amssymb}
\usepackage{graphicx}
\usepackage{url}

\newcommand{\bea}{\begin{eqnarray}}
\newcommand{\eea}{\end{eqnarray}}
\newcommand{\ba}{\begin{array}}
\newcommand{\ea}{\end{array}}
\newcommand{\edc}{\end{document}}
\newcommand{\bc}{\begin{center}}
\newcommand{\ec}{\end{center}}
\newcommand{\be}{\begin{equation}}
\newcommand{\ee}{\end{equation}}

\def\bc{{\mathbb C}}

\newtheorem{thm}{Theorem}[section]
\newtheorem{lem}[thm]{Lemma}

\newtheorem{prop}[thm]{Proposition}
\newtheorem{defin}[thm]{Definition}
\theoremstyle{remark}
\newtheorem{rem}{Remark}[section]

\date{\today}
\begin{document}
\title[Paramagnetic phase]
{Exact description of paramagnetic and ferromagnetic phases of an
Ising model on a third-order Cayley tree}
\author{Hasan Ak\i n}
\address{Hasan Ak\i n, Ceyhun
Atuf Kansu Caddesi 1164. Sokak, 9/4, TR06105, \c{C}ankaya, Ankara,
Turkey} \email{{\tt akinhasan25@gmail.com}}
\date{\today }

\begin{abstract}
In this paper we analytically study the recurrence equations of an
Ising model with three competing interactions on a Cayley tree of
order three. We exactly describe paramagnetic and ferromagnetic
phases of the Ising model. We obtain some rigorous results:
critical temperatures and curves, number of phases, partition
function. Ganikhodjaev et al. \cite{GAUT2011Chaos} have
numerically studied the Ising model on a second-order Cayley tree.
We compare the numerical results to exact solutions of mentioned
model.
\\
\textbf{Keywords}: Cayley tree, Ising model, paramagnetic phase.\\
\textbf{PACS}: 05.70.Fh;  05.70.Ce; 75.10.Hk.
\end{abstract}
\maketitle

\tableofcontents
\section{Introduction}
A phase diagram of a model with variety of transition lines and
modulated phases are obtained via the presence of competing
interactions in magnetic and ferroelectric systems
\cite{Inawashiro}. In order to describe all phases of considered
model, one iterates the recurrence relations associated to the
given Hamiltonian and observes their behavior after a large number
of iterations \cite{Vannimenus,NHSS,MTA1985a,UGAT2012IJMPC}. The
ANNNI (Axial Next-Nearest-Neighbor Ising) model, which consists of
an Ising spin Hamiltonian on a Cayley tree, with ferromagnetic
interactions on the planes, and competing ferromagnetic and
antiferromagnetic interactions between nearest and next-nearest
neighbors along an axial direction, is known to reproduce some
features of these complex phase diagrams
\cite{Lebowitz1,Vannimenus,MTA1985a}.

In recent times, existence and quantification of the phase
diagrams after a large number of iterations of relevant recurrence
equations as a way of probing  ground phases of Ising model have
gained much attention
\cite{GAUT2011Chaos,GU2011a,GTA,Vannimenus,NHSS1,AUT2010AIP,UA2011CJP,UA2010PhysicaA,UGAT2012ACTA,UGAT2012IJMPC,MTA1985a,Inawashiro,Inawashiro-T1983}.
In generally by analyzing the regions of stability of different
types of fixed points of the system of recurrent relations, many
authors have plotted the phase diagrams of the model
\cite{GAUT2011Chaos,GU2011a,GTA,Vannimenus,NHSS1,AUT2010AIP,Moraal,UGAT2012IJMPC,MTA1985a}.
Most results of the above-mentioned works are obtained
numerically. In this paper, we try to get some of these results by
an analytical way by comparing the numerical results.

In \cite{UGAT2012IJMPC}, we numerically study Lyapunov exponent
and modulated phases for the Ising system with competing
interactions on an arbitrary order Cayley tree, the variation of
the wavevector $q$ with temperature in the modulated phase and the
Lyapunov exponent associated with the trajectory of our iterative
system are studied in detail. In \cite{RAU}, we have analytically
study the recurrence equations of an Ising model with two
competing interactions on a Cayley tree of order two and obtained
some exact results: critical temperatures and curves, number of
phases, partition function without considering the numerical
investigation. In \cite{Akin-Saygili2015} we have described the
exact solution of a phase transition problem by means Gibbs
measures of the Potts model on a Cayley tree of order three with
competing interactions. In \cite{ART} we have constructed a class
of new Gibbs measures by extending the known Gibbs measures
defined on a Cayley tree of order $k_0$ to a Cayley tree of higher
order $k > k_0$ for the Ising model. Nazarov and Rozikov
\cite{Nazarov-Rozikov} find the operator corresponding to the
periodic Gibbs distributions with period two and determine the
invariant subsets of this operator, which are used to describe the
periodic Gibbs distributions. Here in order to obtain the periodic
fixed points associated with the recurrence equations, we use the
similar methods in \cite{AGUT,AGTU2013ACTA}.


In the ref. \cite{Akin2017}, the author studies the Gibbs measures
associated with Vannimenus-Ising model for compatible conditions,
he has been interested in the existence of the
translation-invariant Gibbs measures with respect to the
compatible conditions. Our present results differ from
\cite{Akin2016,Akin2017,Akin2017a,AkinT2011CMP}, because we have
considered external magnetic field in the mentioned papers (see
\cite{BRZ,BleherG}).

In the present paper, we analytically study the recurrence
equations of an Ising model with three competing interactions on a
three order Cayley tree. We obtain the paramagnetic, ferromagnetic
and 2-period phases of the model via the related recurrence
equations. We exactly describe paramagnetic phase of the Ising
model. We obtain some rigorous results: critical temperatures and
curves, number of phases, partition function. This model was
numerically studied by Ganikhodjaev \emph{et al}.
\cite{GAUT2011Chaos} on semi-infinitive second-order Cayley tree.
\section{Preliminary}
\subsection{Cayley tree}
A tree  is a graph which is connected and contains no circuits. In
this paper, we consider an semi-infinite tree which has uniformly
bounded degrees. That is, the numbers of neighbors of any vertices
in this tree are uniformly bounded; we call it the uniformly
bounded tree. If the root of a tree has $k$ neighboring vertices
and other vertices have $k+1$ neighboring vertices, we call this
type of tree a Cayley tree. It is easy to see that this type of
tree is the special case of uniformly bounded tree. Cayley trees
(or Bethe lattices) are simple connected undirected graphs $G =
(V, E)$ ($V$ set of vertices, $E$ set of edges) with no cycles (a
cycle is a closed path of different edges), i.e., they are trees
\cite{Akin2017}. Let $\Gamma^k=(V,L,i)$ be the uniform Cayley tree
of order $k$ with a root vertex $x^{(0)}\in V$, where each vertex
has $(k + 1)$ neighbors with $V$ as the set of vertices and the
set of edges. The notation $i$ represents the incidence function
corresponding to each edge $\ell\in L$, with end points
$x_1,x_2\in V$. There is a distance $d(x,y)$ on $V$ the length of
the minimal point from $x$ to $y$, with the assumed length of 1
for any edge (see Figure \ref{cayley-tree-k=3}).

\begin{figure} [!htbp]\label{cayley-tree-k=3}
\centering
\includegraphics[width=60mm]{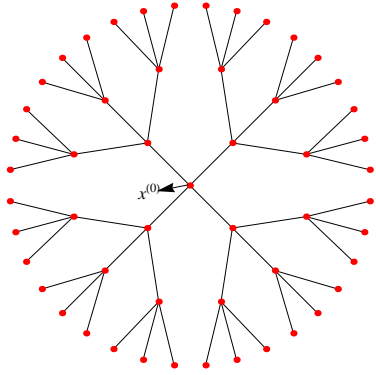}
\caption{Cayley tree of order $k=3$}\label{cayley-tree-k=3}
\end{figure}

Any $k$ ($k >1$)-order Cayley tree $\Gamma^{k}$ is a weave pattern
in which $(k + 1)$ edges from each vertex point extend infinitely
as shown in Fig. \ref{cayley-tree-k=3} ($k=3$). For the Cayley
tree shown as $\Gamma^{k}=(V,\Lambda)$, $V$ denotes the corner
points of the Cayley tree and $\Lambda$ denotes the set of edges.
If there is an edge $\ell$ joining two vertex $x$ and $y$, it is
called "nearest neighbor" and $l = <x, y>$. The distance of $x$
and $y$, over $V$ is defined as $d(x, y)$, the shortest path
between $x$ and $y$. For any given point $x^{(0)}\in V$, the set
of levels according to this point is given as $V_n = \{x \in V
|d(x, x^{(0)})\leq n\}$ and the set of edge points in $V_n$ is
represented as $L_n$.

The set of all vertices with distance $n$ from the root  $x^{(0)}$
is called the $n$th level of $\Gamma^{k}$ ann we denote the sphere
of radius $n$ on $V$ by
$$ W_n=\{x\in V:
d(x,x^{(0)})=n \}
$$ and the ball of radius $n$ by
$$ V_n=\{x\in V:
d(x,x^{(0)})\leq n \}.
$$ The set of direct successors of any vertex
$x\in W_n$ is denoted by
$$S_k(x)=\{y\in W_{n+1}:d(x,y)=1 \}.
$$

A Cayley tree $\Gamma^k$ of order $ k\geq 1 $ is an infinite tree,
i.e., a graph without cycles with exactly $ k+1 $ edges issuing
from each vertex. Let denote the Cayley tree as $\Gamma^k=(V,
\Lambda),$ where $V$ is the set of vertices of $ \Gamma^k$,
$\Lambda$ is the set of edges of $ \Gamma^k$.

The distance $d(x,y), x,y\in V$, on the Cayley tree $\Gamma^k$, is
the number of edges in the shortest path from $x$ to $y$. The
fixed vertex $x^{(0)}$ is called the $0$-th level and the vertices
in $W_n$ are called the $n$-th level. For the sake of simplicity
we put $|x|=d(x,x^{(0)})$, $x\in V$.
\begin{defin}\label{neighborhoods}
Hereafter, we will use the following definitions for
neighborhoods.
\begin{enumerate}
\item Two vertices $x$ and $y$, $x,y \in V$ are called {\it
\textbf{nearest-neighbors (NN)}} if there exists an edge
$l\in\Lambda$ connecting them, which is denoted by $l=<x,y>$.
\item The next-nearest-neighbor
vertices $x\in W_n$ and $z\in W_{n+2}$ are called {\it
\textbf{prolonged next-nearest-neighbors (PNNN)}} if $|x|\neq |y|$
and is denoted by $>x,z<$
\item The triple of vertices $x,y,z$ is called {\it \textbf{ternary prolonged next-nearest-neighbors}}
if $x\in W_n,y\in S(x)$ and $z\in S(y)$ ($x\in W_n,y\in W_{n+1}$
and $z\in W_{n+2}$) for some nonnegative integer $n$  and is
denoted by $>x,y,z<$.
\end{enumerate}
\end{defin}
In this paper, we will consider an Hamiltonian with
\textbf{competing nearest-neighbor interactions},
\textbf{prolonged next-nearest-neighbors (PNNN)}  and
\textbf{ternary prolonged next nearest-neighbor interactions}.
Therefore, we can state the Hamiltonian by {\small
\begin{equation}\label{Ham-Ist} H(\sigma )=-J\sum _{
\begin{array}{l}
 <x,y> \\
 y\in S(x)
\end{array}
} \sigma (x)\sigma (y)-J_p\sum _{
\begin{array}{l}
 >x,z< \\
 z\in S^2(x)
\end{array}
} \sigma (x)\sigma (z)-J_t\sum _{
\begin{array}{l}
 >x,y,z< \\
 y\in S(x) \\
 z\in S(y)
\end{array}
} \sigma (x)\sigma (y)\sigma (z),
\end{equation}}
where $J,J_p, J_t\in \mathbb{R}$ are coupling constants and $<x,
y>$ stands for NN vertices, $>x,z<$ stands for prolonged NNN and
$>x,y,z<$ stands for prolonged ternary NNN. Note that the joule,
symbol $J$, is a derived unit of energy in the International
System of Units.
\section{Recursive equations for Partition Functions}
There are several approaches to derive equation or system
equations describing limiting Gibbs measure (phases) for lattice
models on Cayley tree. One approach is based on properties of
Markov random fields on Bethe lattices
\cite{AT1,Rozikov,Bleher-Zalys,Bleher1990a}. Another approach is
based on recursive equations for partition functions (for example
\cite{Kindermann}). Naturally, both approaches lead to the same
equation (see \cite{AT1}). The second approach is more suitable
for models with competing interactions.

After specifying a Hamiltonian $H$, the equilibrium state of a
physical system with Hamiltonian $H$ is described by the
probability measure
$$
\mu(\sigma _n)=\frac{\exp [-\beta H(\sigma _n)]}{\sum_{\sigma
_n\in \{-1,+1\}^{V_n}} \exp [-\beta H(\sigma _n)]},
$$
where $\beta$ is a positive number which is proportional to the
inverse of the absolute temperature. The above $\mu$ is called the
Gibbs distribution relative to $H$. The standard approach consists
in writing down recurrence equations relating the partition
function
$$
Z_n=\sum _{\sigma _n\in \{-1,+1\}^{V_n}} \exp [-\beta H(\sigma
_n)],
$$
of an $n$-generation tree to the partition function $Z_{n-1}$ of
its subsystems containing $(n-1)$ generations (see Figure
\ref{fig2}).

As usual, we can introduce the notions of ground states (Gibbs
measure) of the Ising model with competing interactions on the
Cayley tree \cite{Akin2016,Akin2017,MAKfree2017}. It is convenient
to use a shorter notation to write down the recurrence system
explicitly:

\begin{eqnarray*}\label{rec-eq2}
\left\{
\begin{array}{l}
 z_1=Z_N\left(
\begin{array}{ccc}
 + & + & + \\
  & + &
\end{array}
\right), \\
 z_2=Z_N\left(
\begin{array}{ccc}
 + & + & - \\
   & + &
\end{array}
\right)=Z_N\left(
\begin{array}{ccc}
 + & - & + \\
   & + &
\end{array}
\right)=Z_N\left(
\begin{array}{ccc}
 - & + & + \\
   & + &
\end{array}
\right), \\
 z_3=Z_N\left(
\begin{array}{ccc}
 + & - & - \\
  & + &
\end{array}
\right)=Z_N\left(
\begin{array}{ccc}
 - & + & - \\
   & + &
\end{array}
\right)=Z_N\left(
\begin{array}{ccc}
 - & - & + \\
 & + &
\end{array}
\right), \\
 z_4=Z_N\left(
\begin{array}{ccc}
 - & - & - \\
  & + &
\end{array}
\right),
\end{array}
\right.
\end{eqnarray*}
\begin{eqnarray*}\label{rec-eq2}
\left\{
\begin{array}{l}
 z_5=Z_N\left(
\begin{array}{ccc}
 + & + & + \\
  & - &
\end{array}
\right), \\
 z_6=Z_N\left(
\begin{array}{ccc}
 + & + & - \\
   & - &
\end{array}
\right)=Z_N\left(
\begin{array}{ccc}
 + & - & + \\
   & - &
\end{array}
\right)=Z_N\left(
\begin{array}{ccc}
 - & + & + \\
   & - &
\end{array}
\right),\\
 z_7=Z_N\left(
\begin{array}{ccc}
 + & - & - \\
  & - &
\end{array}
\right)=Z_N\left(
\begin{array}{ccc}
 - & + & - \\
   & - &
\end{array}
\right)=Z_N\left(
\begin{array}{ccc}
 - & - & + \\
 & - &
\end{array}
\right). \\
 z_8=Z_N\left(
\begin{array}{ccc}
 - & - & - \\
  & - &
\end{array}
\right).
\end{array}
\right.
\end{eqnarray*}
\begin{figure} [!htbp]\label{fig2}
\centering
\includegraphics[width=80mm]{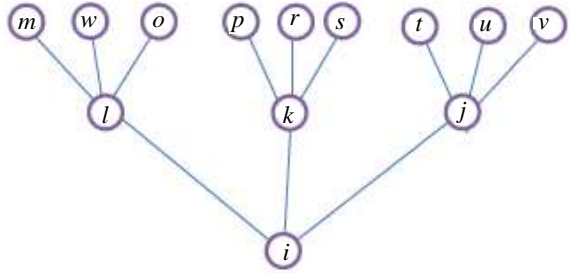}\ \ \ \ \ \ \ \ \ \ \ \
\caption{Configurations on semi-finite Cayley tree of order three
with levels 2. Schematic diagram to be illustrated the summation
used in Equation \eqref{partition1aa}.}\label{fig2}
\end{figure}
For the sake simplicity, denote $Z_{(n+1)}\left(
\begin{array}{ccc}
 l & k & j \\
  & i &
\end{array}
\right)=Z_{(n+1)}\left(i;l,k,j\right)$. We have
\begin{eqnarray}\label{partition1aa}
Z_{n+1}\left(i;l, k, j\right)&=&\sum _{m,w,o,p,r,s,t,u,v\in
\{-1,+1\}} [\exp
(A(m,w,o,p,r,s,t,u,v))\\
&&Z_{n}\left(l;m,w,o\right)Z_{n}\left(k;p,r,s\right)Z_{n}\left(j;t,u,v\right)],\nonumber
\end{eqnarray}
where
\begin{eqnarray*}
A(m,w,o,p,r,s,t,u,v)&=&Ji(j+k+l)+J_pi(m+w+o+p+r+s+t+u+v)\\&&+J_t(i(l(m+w+o)+k(p+r+s)+j(t+u+v)),
\end{eqnarray*}
$i,j,k,l\in \{-1,+1\}$ and $n=1,2,\ldots $ (see Figure
\ref{fig2}). Therefore, we can calculate the partial partition
functions as follows
$$
\left\{
\begin{array}{l}
 z_1^{(n+1)}=a^3(bc)^{-9}\left((bc)^6z_1^{(n)}+3(bc)^4z_2^{(n)}+3(bc)^2z_3^{(n)}+z_4^{(n)}\right){}^3 \\
 z_2^{(n+1)}=a(bc)^{-9}\left((bc)^6z_1^{(n)}+3(bc)^4z_2^{(n)}+3(bc)^2z_3^{(n)}+z_4^{(n)}\right){}^2 \\
 \ \ \ \ \ \ \  \ \ \text{}\times \left(b^6z_5^{(n)}+3c^2b^4z_6^{(n)}+3b^2c^4z_7^{(n)}+c^6z_8^{(n)}\right) \\
 z_3^{(n+1)}=a^{-1}(bc)^{-9}\left((bc)^6z_1^{(n)}+3(bc)^4z_2^{(n)}+3(bc)^2z_3^{(n)}+z_4^{(n)}\right) \\
 \ \ \ \ \ \ \  \ \ \text{}\times \left(b^6z_5^{(n)}+3c^2b^4z_6^{(n)}+3b^2c^4z_7^{(n)}+c^6z_8^{(n)}\right){}^2 \\
 z_4^{(n+1)}=a^{-3}(bc)^{-9}\left(b^6z_5^{(n)}+3c^2b^4z_6^{(n)}+3b^2c^4z_7^{(n)}+c^6z_8^{(n)}\right){}^3
\end{array}
\right.
$$
$$
\left\{
\begin{array}{l}
 z_5^{(n+1)}=a^{-3}(bc)^{-9}\left(z_1^{(n)}+3(bc)^2z_2^{(n)}+3(bc)^4z_3^{(n)}+(bc)^6z_4^{(n)}\right)^3 \\
 z_6^{(n+1)}=a^{-1}(bc)^{-9}\left(z_1^{(n)}+3(bc)^2z_2^{(n)}+3(bc)^4z_3^{(n)}+(bc)^6z_4^{(n)}\right)^2 \\
  \ \ \ \ \ \ \  \ \ \text{}\times \left(c^6z_5^{(n)}+3c^4b^2z_6^{(n)}+3c^2b^4z_7^{(n)}+b^6z_8^{(n)}\right) \\
 z_7^{(n+1)}=a(bc)^{-9}\left(z_1^{(n)}+3(bc)^2z_2^{(n)}+3(bc)^4z_3^{(n)}+(bc)^6z_4^{(n)}\right) \\
  \ \ \ \ \ \ \  \ \ \text{}\times \left(c^6z_5^{(n)}+3c^4b^2z_6^{(n)}+3c^2b^4z_7^{(n)}+b^6z_8^{(n)}\right)^2 \\
 z_8^{(n+1)}=a^3(bc)^{-9}\left(c^6z_5^{(n)}+3c^4b^2z_6^{(n)}+3c^2b^4z_7^{(n)}+b^6z_8^{(n)}\right)^3,
\end{array}
\right.
$$
where $a=e^{\beta J}$, $b=e^{\beta J_p}$, $c=e^{\beta J_t}$.
Noting that
\begin{eqnarray*}
&&\left(z_2^{(n+1)}\right)^3=\left(z_1^{(n+1)}\right)^2z_4^{(n+1)},\
\left(z_3^{(n+1)}\right)^3=z_1^{(n+1)}\left(z_4^{(n+1)}\right)^2,\\
&&\left(z_6^{(n+1)}\right)^3=\left(z_5^{(n+1)}\right)^2z_8^{(n+1)},\
\left(z_7^{(n+1)}\right)^3=z_5^{(n+1)}\left(z_8^{(n+1)}\right)^2,
\end{eqnarray*}
only four independent variables remain, and through the
introduction of the new variables
$u_i^{(n+1)}=(z_2^{(n+1)})^{\frac{1}{3}}$, we can obtain the
recurrence system the following simpler form:

\begin{eqnarray}\label{rec-eq2}
\left\{
\begin{array}{l}
 u_1^{(n+1)}=\frac{a}{(bc)^3}\left((bc)^2u_1^{(n)}+u_4^{(n)}\right)^3 \\
 u_4^{(n+1)}=\frac{1}{a(bc)^3}\left(b^2u_5^{(n)}+c^2u_8^{(n)}\right)^3 \\
 u_5^{(n+1)}=\frac{1}{a(bc)^3}\left(u_1^{(n)}+(bc)^2u_4^{(n)}\right)^3 \\
 u_8^{(n+1)}=\frac{a}{(bc)^3}\left(c^2u_5^{(n)}+b^2u_8^{(n)}\right)^3.
\end{array}
\right.
\end{eqnarray}
Let us define  operator as
$$F: u^{(n)} = (u_1^{(n)},
u_4^{(n)}, u_5^{(n)}, u_8^{(n)})\in \mathbb{R}^{4}_+Б\rightarrow
F(u^{(n)}) = (u_1^{(n+1)}, u_4^{(n+1)}, u_5^{(n+1)},
u_8^{(n+1)})\in \mathbb{R}^{4}_+.
$$
Then we can write the recurrence equations \eqref{rec-eq2} as
$u^{(n+1)} = F(u^{(n)}), n > 0$ which in the theory of dynamical
systems is called a trajectory of the initial point $u^{(0)}$
under the action of the operator $F$. Thus we can determine the
asymptotic behavior of $Z_n$ for $n\rightarrow \infty$ by the
values of $\lim_{n\rightarrow \infty } u^{(n)}$ i.e., the
trajectory of $u^{(0)}$ under the action of the operator $F$. In
this paper we study the trajectory (dynamical system) for a given
initial point $u^{(0)}\in \mathbb{R}^{4}_+$.

\section{Dynamics of the operator $F$}\label{Dynamics of the
operator}
\subsection{Fixed Points of the operator $F$}
In this subsection we are going to determine the fixed points,
i.e., solutions to $F(u) = u$. Denote $Fix(F ) = \{u : F(u) =
u\}.$\\
We introduce the new variables $\alpha =\sqrt[3]{a}$, $v_i =
\sqrt[3]{u_i}$, for $i = 1, 4, 5, 8$. Then the equation $F(v) = v$
becomes
\begin{eqnarray}\label{rec-eq3}
\left\{
\begin{array}{l}
 v_1=\frac{\alpha}{(bc)}\left((bc)^2v_1^{3}+v_4^{3}\right) \\
 v_4=\frac{1}{\alpha (bc)}\left(b^2v_5^{3}+c^2v_8^{3}\right) \\
 v_5=\frac{1}{\alpha(bc)}\left(v_1^{3}+(bc)^2v_4^{3}\right)\\
 v_8=\frac{\alpha}{(bc)}\left(c^2v_5^{3}+b^2v_8^{3}\right).
\end{array}
\right.
\end{eqnarray}
\subsection{The fixed points of the operator $F$ for $c=1$}
Let us consider the following set:
\begin{eqnarray}\label{set-invariant}
A=\{(v_1, v_4, v_5, v_8)\in
\mathbb{R}^{4}_+:v_1=v_8,v_4=v_5,c=1\}.
\end{eqnarray}
Note that the set $A$ is invariant with respect to $F$ i.e.,
$F(A)\subset БA.$
\begin{lem}\label{Lemma-fixed-set}
If a vector $\overrightarrow{u}$ is a fixed point of the operator
$F$, then $\overrightarrow{u}\in
M_1:=\{\overrightarrow{u}=(u_1,u_4,u_5,u_8)\in \mathbf{R}^{4}_{+}:
u_1=u_8,u_4=u_5\}$ or $\overrightarrow{u}\in
M_2:=\{\overrightarrow{u}=(u_1,u_4,u_5,u_8)\in \mathbf{R}^{4}_{+}:
\left(\sqrt[3]{u_4}+\sqrt[3]{u_5}\right)^2-\sqrt[3]{u_4u_5}=\psi
\left(\left(\sqrt[3]{u_1}+\sqrt[3]{u_8}\right){}^2-\sqrt[3]{u_1u_8}\right)\}$,
where $\psi(y)
=\frac{b^3+\alpha  y (1-b^4)}{b^2 \alpha
\left(b \alpha y-1\right)}$.
\end{lem}

\begin{proof} From the system \eqref{rec-eq3}, we have
\begin{eqnarray}\label{fixed-1a}
&& \left(v_1-v_8\right)\left(b\alpha
\left(v_1^2+v_1v_8+v_8^2\right)-1\right)+\alpha
b^{-1}\left(v_4-v_5\right)\left(v_4^2+v_4v_5+v_5^2\right)=0
\\\label{fixed-1b}
&& \left(v_1-v_8\right)(\alpha
b)^{-1}\left(v_1^2+v_1v_8+v_8^2\right)+\left(v_4-v_5\right)(\alpha
b)^{-1}\left(b^2v_4^2+b^2v_4v_5+b^2v_5^2+\alpha [b]\right)=0.
\end{eqnarray}
From \eqref{fixed-1a} and \eqref{fixed-1b} one can conclude that
if $v_1=v_8$ (respectively $v_3=v_4$), then $v_3=v_4$
(respectively $v_1=v_8$). Therefore, $v_3=v_4$ if and only if
$v_3=v_4$.

Now let us assume that $v_1\neq v_8$ and $v_3\neq v_4$, then we
can reduce the equations \eqref{fixed-1a} and \eqref{fixed-1b} to
the following equation:
$$
\frac{\alpha  b^{-1}\left(v_4^2+v_4v_5+v_5^2\right)}{\left(b
\alpha
\left(v_1^2+v_1v_8+v_8^2\right)-1\right)}=\frac{\left(b^2v_4^2+b^2v_4v_5+b^2v_5^2+\alpha
b\right)}{\left(v_1^2+v_1v_8+v_8^2\right)}.
$$
Therefore, from the last equation we have
\begin{equation}\label{fixed-inv2}
v_4^2+v_4v_5+v_5^2=\frac{b^3+\alpha  \left(v_1^2+v_1
v_8+v_8^2\right)(1-b^4)}{b^2 \alpha  \left(b \alpha
\left(v_1^2+v_1 v_8+v_8^2\right)-1\right)}.
\end{equation}
The equation \eqref{fixed-inv2} gives
$\left(\sqrt[3]{u_4}+\sqrt[3]{u_5}\right)^2-\sqrt[3]{u_4u_5}=\psi
\left(\left(\sqrt[3]{u_1}+\sqrt[3]{u_8}\right){}^2-\sqrt[3]{u_1u_8}\right)$
(see \cite{GU2011a}).
\end{proof}
\begin{rem}
The fixed points of the operator $F$ belonging to the set $M_2$
give the ferromagnetic phases corresponding to the Ising model
\eqref{Ham-Ist}. In order to examine the fixed points of the
operator $F$ belonging to the set $M_2$ is analytically very
difficult. The ferromagnetic phase regions corresponding to the
Ising model \eqref{Ham-Ist} can numerically be determined.
\end{rem}
\subsection{The existence of paramagnetic and ferromagnetic phases}\label{The paramagnetic phases}
Let us first study the fixed points of the operator $F$ which
belong in $A$ given in the equation \eqref{set-invariant}. The
fixed points of the operator $F$ determine the paramagnetic phases
corresponding to the Ising model \eqref{Ham-Ist}.

The condition $u_1=u_8, u_4 = u_5$ reduces the equation $F(u) = u$
to the following equation
\begin{eqnarray}\label{c=1-fixed points}
x=g(x):=\alpha^2\left(\frac{1+b^2x^3}{b^2+x^3}\right),
\end{eqnarray}
where $x=\frac{v_1}{v_4}.$

The following Proposition gives full description of positive fixed
points of the function g in \eqref{c=1-fixed points}.
\begin{prop}\label{proposition1}
The equation \eqref{c=1-fixed points} (with $x \geq 0, \alpha > 0,
b > 0$) has one solution if $b<1$.  Assume that If $b
> \sqrt{2}$, then  the equation \eqref{c=1-fixed points} has 2 solutions
if either $\eta_1(b)=\alpha^{-2}$ or $\eta_2(b)=\alpha^{-2}$. If
$\eta_1(b)<\alpha^{-2}<\eta_2(b)$ and then there exists
$\eta_1(b)$, $\eta_2(b)$ with $0<\eta_1(b)<\eta_2(b)$ such that
the equation \eqref{c=1-fixed points} has 3  solutions. In this
case, we have
\begin{eqnarray*}
\eta_i
(b)&=&\frac{1}{x_i}\left(\frac{1+b^2x_i^3}{b^2+x_i^3}\right),
\end{eqnarray*}
where $x_i$ are the solutions of the equation $ b^2
x^6-2\left(b^4-2\right)x^3+b^2=0.$.
\end{prop}
\begin{proof}
Let us take the first and the second derivatives of the function
$g$, we have
\begin{eqnarray}\label{first-derivative}
g'(x)=\frac{3\left(b^4-1\right) x^2
\alpha^2}{\left(b^2+x^3\right)^2},
\end{eqnarray}
$$
g''(x)=\frac{6 \alpha ^2 \left(b^4-1\right) x \left(b^2-2
x^3\right) }{\left(b^2+x^3\right)^3}.
$$
From \eqref{first-derivative}, if $b<1$ (with $x \geq 0$) then $g$
is decreasing and there can only be one solution of $g(x)=x.$
Thus, we can restrict ourselves to the case in which $b>1.$  It is
obvious that the graph of $y=g(x)$ over interval
$(0,\sqrt[3]{\frac{b^2}{2}})$ is concave up  and the graph of
$y=g(x)$ over interval $(\sqrt[3]{\frac{b^2}{2}},\infty)$ is
concave down. As a result, there are at most 3 positive solutions
for $g(x)=x.$

According to Preston \cite[Proposition 10.7]{Preston}, there can
be more than one fixed point of the function $g$ if and only if
there is more than one solution to $xg'(x) = g(x)$, which is the
same as
$$
b^2 x^6-2\left(b^4-2\right)x^3+b^2=0.
$$
These roots are
$$
x_1=\sqrt[3]{\frac{-2+b^4-\sqrt{4-5
b^4+b^8}}{b^2}},x_2=\sqrt[3]{\frac{-2+b^4+\sqrt{4-5
b^4+b^8}}{b^2}}.
$$
Note that if $4-5 b^4+b^8\geq 0$, then the roots $x_1$ and $x_2$
are real numbers. In this case, $4-5 b^4+b^8\geq 0$ if and only if
$b\in (0,1)\cap(\sqrt{2},\infty )$.
\end{proof}
\subsubsection{An Illustrative example for $c=1$.}
We have $P_4(x) = 0$ with a polynomial
\begin{equation}\label{polinomial4d}
P_4=x^4-\alpha ^2b^2x^3+b^2x-\alpha^2
\end{equation}
the coefficients of which depend on parameters $\alpha, b $. Thus
we obtain a quartic equation. Such equations can be solved using
known formulas (see \cite{Wolfram}), since we will have some
complicated formulas for the coefficients and the solutions, we do
not present the solution here. Nonetheless, we have manipulated
the polynomial equation via Mathematica \cite{Wolfram}. Here we
will only deal with positive fixed points, because of the
positivity of exponential functions.
We have obtained at most 3 positive real roots for some parameters
$J$, $J_p$ and $J_{t}=0$ (coupling constants) and temperature $T$.
\begin{figure} [!htbp]\label{3parametric-phase-c=1}
\centering
\includegraphics[width=60mm]{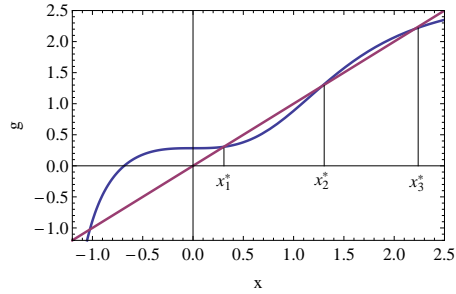}\label{fig1c}
\caption{(Color online) There exist three positive roots of the
equation \eqref{c=1-fixed points} for $J = -4.5, J_p = 14, T =
24.6$.}\label{3parametric-phase-c=1}
\end{figure}
For example, Fig. \ref{3parametric-phase-c=1} shows that there are
3 positive fixed points of the function \eqref{c=1-fixed points}
for $J = -4.5, J_p = 14, J_{t}=0, T = 24.6$. These all fixed
points are $x_0^*= -1.02554,x_1^*=0.306205, x_2^*=1.28008,
x_3^*=2.20209$, respectively. It is clear that $x_1^*=0.306205,
x_3^*=2.20209$ are stable, and $x_2^*=1.28008$ is unstable. The
point $x_{cr}=1.15992$ is the breaking point of  the function $g$.
There are two extreme paramagnetic phases associated to the
positive fixed points.

\subsection{Periodic Points of the operator $F$}
One of the most interested problems in the investigation of non
linear dynamical systems is the existence of periodic points.
While, for the one-dimensional case, every non linear dynamical
systems contains periodic points there is a $d$-dimensional
($d>1$) which contains no periodic points. In statistical physics,
these periodic points reveal the phase types corresponding to the
given model. We recall some definitions and results first.
\begin{defin}
A point $\textbf{u}=(u_1,u_4,u_5,u_8)$ in $\mathbf{R}^{4}_+$ is
called a periodic point of $F$ if there exists $p$ so that
$F^{p}(\textbf{u}) = \textbf{u}$ where $F^{p}$ is the $p$th
iterate of $F$. The smallest positive integer $p$ satisfying the
above is called the prime period or least period of the point
$\textbf{u}$. Denote by Per$_p(F)$ the set of periodic points with
prime period $p$.
\end{defin}
Let us first describe periodic points with $p=2$ on $M_1$ in this
case the equation $F(F(\textbf{u})) = \textbf{u}$ can be reduced
to a description of 2-periodic points of the function  $g$ defined
in \eqref{c=1-fixed points} i.e., to a solution of the equation
\begin{equation}\label{2-period1}
g (g(x)) = x.
\end{equation}
Note that the fixed points of f are solutions to
\eqref{2-period1}, to find other solutions we consider the
equation
$$
\frac{g(g(x))-x}{g(x)-x}=0,
$$
simple calculations show that the last equation is equivalent to
the following
\begin{eqnarray}\label{polinomial6d}
&&p_6(x):=4b^2(1+b^4 \alpha ^6)x^6+ \alpha ^2(b^4-1)x^5+b^2 \alpha
^4(b^4-1)x^4 + 2 b^4(1+ \alpha^6)x^3\\\nonumber &&+ b^2 \alpha
^2(b^4-1)x^2+\alpha ^4(b^4-1)x+b^2(b^4+\alpha ^6)=0.
\end{eqnarray}
In order to describe the periodic points with $p=2$ on $M_1$ of
the operator $F$, we should find the solutions to
\eqref{polinomial6d} which are different from the solutions of the
equation \eqref{polinomial4d}. On the other words, we obtain the
set
\begin{equation}\label{2-period-set}
M_3:=\{(u_1,u_4,u_5,u_8)\in \mathbf{R}^{4}_+:g (g(x)) = x
\}.
\end{equation}
Therefore, we should examine the roots of the polynomial $p_6(x)$
of degree 6. As mentioned above, the roots of such polynomials can
be described using known formulas. Since some complicated formulas
for the coefficients and the solutions are included, we will not
present the solution here. In order to illustrate the problem, we
have manipulated the equation \eqref{2-period1} via Mathematica
\cite{Wolfram} (see Figure \ref{periodic2-phase} (red color)). The
black graph in the Figure \ref{periodic2-phase} represents the
roots of the nonlinear function $y=g(x).$
\begin{figure} [!htbp]\label{periodic2-phase}
\centering
\includegraphics[width=60mm]{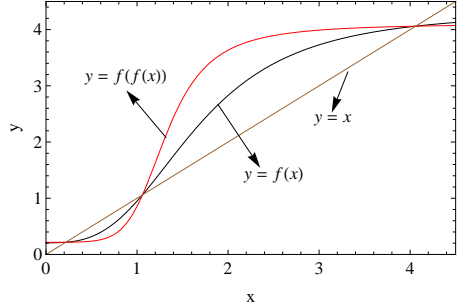}\label{periodic2-phase}
\caption{(Color online) There exist three positive roots of the
equation \eqref{2-period1} (black color) for $J = -6, J_p = 56.9,
T = 75$. Also, the equation \eqref{2-period1} has three positive
roots (red color) for $J = -6, J_p = 56.9, T =
75$.}\label{periodic2-phase}
\end{figure}

We can obtain an initial point of the sequence
$(u_1^{(n)},u_4^{(n)},u_5^{(n)},u_8^{(n)})$ under positive
boundary condition as follows:
$$
\textbf{u}_0=(u_1^{(0)},u_4^{(0)},u_5^{(0)},u_8^{(0)})=(ab^3c^3,\frac{b^3}{ac^3},\frac{1}{ac^3b^3},\frac{ac^3}{b^3}).
$$
In order to study some useful features of the function $g$, let us
give the following lemma.

\begin{lem}\label{repelling-point}
1) If $b> 1$ then the sequence $x_n =(x_{n-1})$, $n = 1, 2, . . .$
converges for the initial point
$x_0=\sqrt[3]{\frac{u_1^{(0)}}{u_4^{(0)}}}=\alpha ^{2}
> 0$  under positive boundary condition,  where $g$ is defined in \eqref{c=1-fixed points}.

2) If $b<1$ then the sequence $y_n =f(y_{n-1})$, $n = 1, 2, . . .$
converges for the initial point $y_0=\frac{\alpha ^2 \left(1+b^2
\alpha ^{6}\right)}{b^2+\alpha ^{6}}>0$ under positive boundary
condition, where $f(x) = g(g(x)).$
\end{lem}
\begin{proof} 1) For $b> 1$ we have
$g'(x)=\frac{3\left(b^4-1\right) x^2
\alpha^2}{\left(b^2+x^3\right)^2}>0$ i.e., $g$ is an increasing
function. Here we consider the case when the function $g$ has
three fixed points $x^{*}_i , i = 1, 2, 3$ (see Proposition
 \ref{proposition1}  and Figure \ref{3parametric-phase-c=1}). We have that the point
 $x^{*}_2$ is a repeller i.e., $g'(x^{*}_2)> 1$ and the points $x^{*}_1,x^{*}_3$
are attractive i.e., $g'(x^{*}_1)< 1$ and $g'(x^{*}_3)< 1$. Now we
shall take arbitrary $x_0 > 0$ and prove that $x_n = g(x_{n-1})$,
$n\geq 1$ converges as $n\rightarrow \infty$. For any $x \in (0,
x^{*}_1)$ we have $x < g(x)<x^{*}_1$, since $g$ is an increasing
function, from the last inequalities we get $x < g(x) < g^{2}(x) <
g(x^{*}_1) = x^{*}_1$. Iterating this argument we obtain
$g^{n-1}(x) < g^{n}(x)<x^{*}_1$, which for any $x_0 \in (0,
x^{*}_1)$ gives $x_{n-1} < x_n < x^{*}_1$ i.e., $x_n$ converges
and its limit is a fixed point of $g$, since $g$ has a unique
fixed point $x^{*}_1$ in $(0, x^{*}_1]$. We conclude that the
limit is $x^{*}_1$. For $x\in (x^{*}_1, x^{*}_2)$ we have
$x^{*}_2> x> g(x)> x^{*}_1$, consequently $x_n > x_{n+1}$ i.e.,
$x_n$ converges and its limit is again $x^{*}_1$. Similarly, one
can show that if $x_0 > x^{*}_2$ then $x_n\rightarrow x^{*}_3$ as
$n\rightarrow \infty$.

2) For $b < 1$ we have $g$ is decreasing and has a unique fixed
point $x_1$ which is repelling, but $f$ is increasing since $f'(x)
= g'(g(x))g'(x)>0$. We have that $f$ has at most three fixed
points (including $x_1$). The point $x_1$ is repelling for $f$
too, since $f'(x_1) = g'(g(x_1))g'(x_1)= (g'(x_1))^2 > 1$. But
fixed points $x_-, x_+$ of $f$ are attractive. Hence one can
repeat the same argument of the proof of the part 1) for the
increasing function $f$ and complete the proof.
\end{proof}

\subsection{The phase diagrams of the model}
For plotting of  the phase diagrams in the Hamiltonian
three-parameter spaces, the following choice of reduced variables
is convenient:
\begin{eqnarray}\label{eq8}
x^{(n)} =\frac{u_{4}^{(n)}+u_{5}^{(n)}}{u_1^{(n)}+u_{8}^{(n)}}, \
y^{(n)} = \frac{u_1^{(n)}-u_{8}^{(n)}}{u_1^{(n)}+u_{8}^{(n)}}, \
z^{(n)} = \frac{u_{4}^{(n)}-u_{5}^{(n)}}{u_1^{(n)}+u_{8}^{(n)}}.
\end{eqnarray}
The variable $x^{(n)}$ is just a measure of the frustration of the
nearest-neighbor bonds and is not an order parameter like
$y^{(n)}, z^{(n)}$. It is convenient to know the broad features of
the phase diagram before discussing the different transitions in
more detail (see \cite{Vannimenus} for details). This can be
achieved numerically in a straightforward fashion.

Let $T/J=\alpha$, $-J_p/J=\beta$, $-J_{t}/J=\gamma$ and
respectively $a=\exp(\alpha^{-1}),
 b=\exp(-\alpha^{-1}\beta)$ and $c=\exp(-\alpha^{-1}\gamma)$. From the equations \eqref{eq8}, we can obtain the following
recurrence dynamical system:
\begin{eqnarray}\label{dynamical system2}
\left\{
\begin{array}{l}
 x^{(n+1)}=\frac{\left(c^2\left(1-y^{(n)}\right)+b^2\left(x^{(n)}-z^{(n)}\right)\right)^3+\left(1+y^{(n)}+b^2c^2\left(x^{(n)}+z^{(n)}\right)\right)^3}
{a^2\left(\left(x^{(n)}+b^2c^2\left(1+y^{(n)}\right)+z^{(n)}\right)^3+\left(b^2\left(1-y^{(n)}\right)+c^2\left(x^{(n)}-z^{(n)}\right)\right)^3\right)}\\ \\
 y^{(n+1)}=\frac{\left(x^{(n)}+b^2c^2\left(1+y^{(n)}\right)+z^{(n)}\right)^3-\left(b^2\left(1-y^{(n)}\right)+c^2\left(x^{(n)}-z^{(n)}\right)\right)^3}
{\left(x^{(n)}+b^2c^2\left(1+y^{(n)}\right)+z^{(n)}\right)^3+\left(b^2\left(1-y^{(n)}\right)+c^2\left(x^{(n)}-z^{(n)}\right)\right)^3}\\\\
z^{(n+1)}=\frac{\left(c^2\left(1-y^n\right)+b^2\left(x^n-z^n\right)\right)^3-\left(1+y^n+b^2c^2\left(x^n+z^n\right)\right)^3}
{a^2\left(\left(x^n+b^2c^2\left(1+y^n\right)+z^n\right)^3+\left(b^2\left(1-y^n\right)+c^2\left(x^n-z^n\right)\right)^3\right)}.
\end{array}
\right.
\end{eqnarray}
The system of three equations finally obtained in \eqref{dynamical
system2} is less complicated than one might have anticipated. It
remains difficult to tackle analytically apart from simple limits
and numerical methods are necessary to study its detailed behavior
(see \cite{Vannimenus}).

Starting from initial conditions
 \begin{eqnarray}\label{initial-con1}
\left\{
\begin{array}{l}
 x^{(1)}=\frac{1}{a^2c^6}, \\
 y^{(1)}=\frac{b^6-1}{b^6+1}, \\
 z^{(1)}=\frac{b^6-1}{a^2c^6(b^6+1)}.
\end{array}
\right.
\end{eqnarray}
that corresponds to positive boundary condition
$\bar{\sigma}^{(n)}(V\setminus V_n)\equiv 1,$ one iterates the
recurrence relations \eqref{rec-eq2} and observes behavior of the
phase diagrams after a large number of iterations ($n=10 000$).
For the fixed points, the corresponding magnetization $m$ is given
by
\begin{equation}
m^{(n)}=\frac{\left(1+x^{(n)}+y^{(n)}+z^{(n)}\right)^3-\left(1+x^{(n)}-y^{(n)}-z^{(n)}\right)^3}
{\left(1+x^{(n)}+y^{(n)}+z^{(n)}\right)^3+\left(1+x^{(n)}-y^{(n)}-z^{(n)}\right)^3}.
\end{equation}
The initial point of the magnetization $m$ can be obtained as;
\begin{equation}\label{initial-mag1}
m^{(1)}=\frac{(b^2-1)\left((1+b^2)^2-b^2\right)\left((1+b^6)^2-b^6\right)}{\left(1+b^2\right)\left(1-b^2+b^4\right)\left(1-b^6+b^{12}\right)}.
\end{equation}
Here the variable $x$ is a measure of the frustration of the
nearest-neighbor bonds \cite{Vannimenus}. Since for a paramagnetic
phase we have $u_1 = u_8$ and $u_4= u_5$, we get $y^{(n)}=
z^{(n)}\rightarrow 0$. Hence $m = 0$, but in case of coexistence
of several paramagnetic phases their measure of the frustration
(i.e. $x$) are different. These different values of $x$ are the
solutions to \eqref{c=1-fixed points}.

Now, assume that $c=\exp(-\alpha^{-1}\gamma)=1$. In the simplest
situation a fixed point $\textbf{u}^{*} = (u^{*}_1, u^{*}_4,
u^{*}_5, u^{*}_8)\in \mathbf{R}^{4}_{+}$ is reached. Possible
initial conditions with respect to  different boundary conditions
can be obtained in \cite{Vannimenus,MTA1985a}. In this paper, we
consider initial conditions \eqref{initial-con1} and
\eqref{initial-mag1}. Depending on $u^{*}_1, u^{*}_4, u^{*}_5,
u^{*}_8$, in the simplest situation a fixed point $(x^*,y^*,z^*)$
is reached. It corresponds to a {\it paramagnetic } phase (briefly
{\bf P}) if $y^*=0,z^*=0 $ or to a {\it ferromagnetic} phase
(briefly {\bf  F}) if $y^*,z^* \neq 0.$
The system may be periodic with period $p$, i.e. the periodic
phase is a configuration with some period. If the case $p=2$
corresponds to {\it antiferromagnetic } phase (briefly {\bf P2})
and the case $p=4$ corresponds to so-called {\it antiphase}
(briefly {\bf  P4}), that denoted \ $<2>$ for compactness in
\cite{Vannimenus,MTA1985a}.

Finally, the system may remain aperiodic, i.e. very long period to
compute or non-periodic. The distinction between a truly aperiodic
case and one with a very long period is difficult to make
numerically. Detailed information about the phase  analysis and
the relation with partition functions, it is mainly refered to
works given by Vannimenus \cite{Vannimenus}, Uguz {\it et al}
\cite{UGAT2012IJMPC} and Mariz {\it et al} \cite{MTA1985a}. Below
we just consider periodic phases with period $p$ where $p\leq 12$
(briefly {\bf P2-P12}).

It corresponds to (see \cite{Vannimenus,MTA1985a,Inawashiro} for
details) the cases in Lemma \ref{Lemma-fixed-set}:
\begin{itemize}
    \item a paramagnetic phase: if $\textbf{u}^{*}\in M_1$, in the figure \ref{fazk=3} (b), the white regions represent the paramagnetic
    phase. This represents  the set $M_1$ given in Lemma \ref{Lemma-fixed-set};
    \item a ferromagnetic phase: if $\textbf{u}^{*}\in M_2$, in the figure \ref{fazk=3} (b), the red regions represent the ferromagnetic
    phase. This represents  the set $M_2$ given in Lemma \ref{Lemma-fixed-set};
    \item in the figure \ref{fazk=3} (b), the yellow regions represent the
    \textbf{P2} phase. This represents  the set $M_3$ given in
    \eqref{2-period-set}.
\end{itemize}
\begin{figure}[!htbp]\centering
\includegraphics[width=50mm]{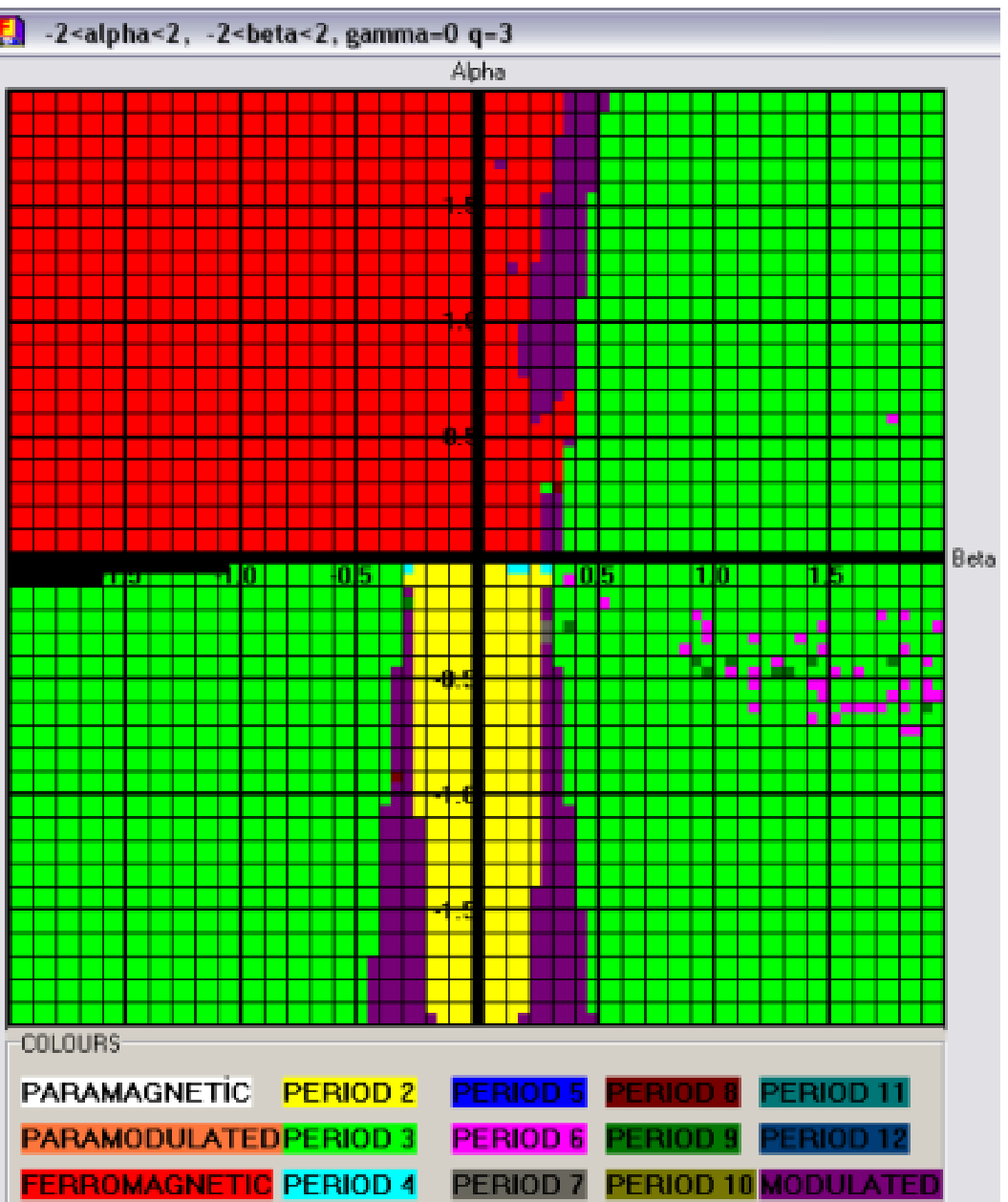}\qquad \quad
\includegraphics[width=50mm]{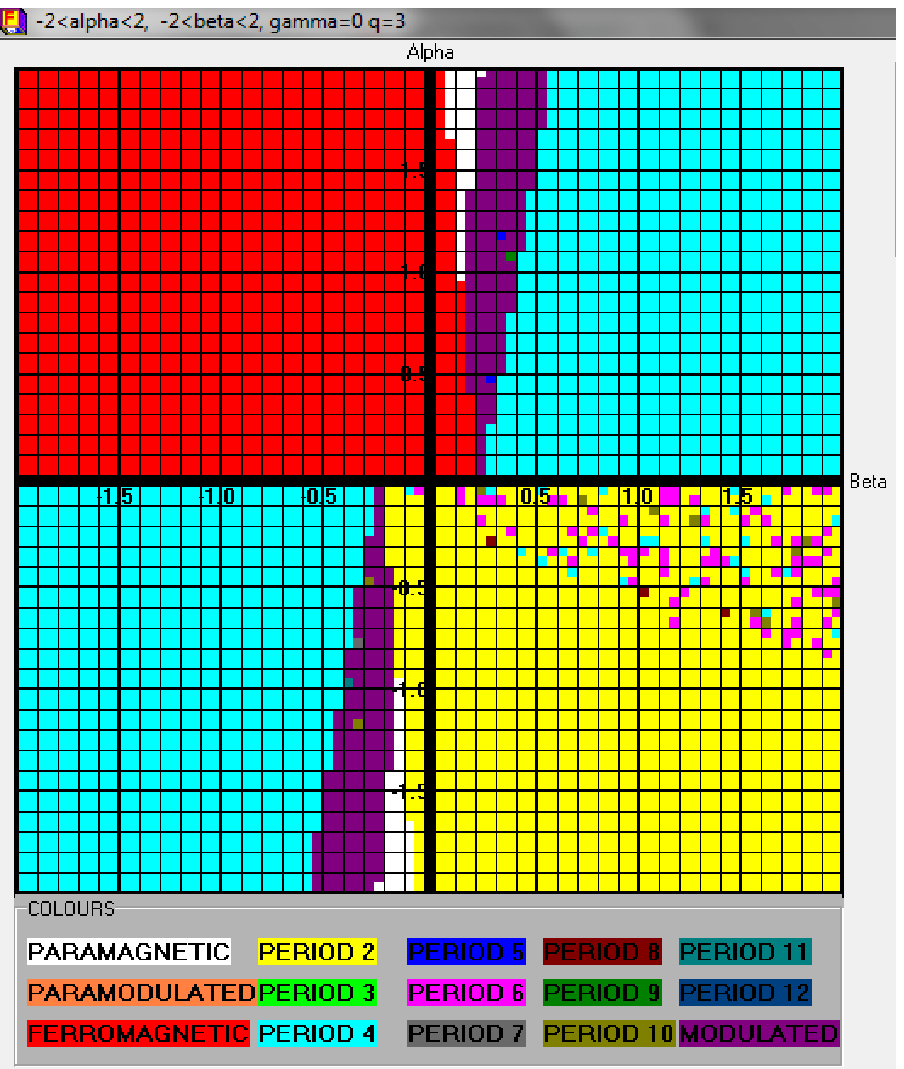}
\caption{(Color online) (a) Phase diagram of the model for
$J_p=0$; (b)  Phase diagram of the model for
$J_t=0$.}\label{fazk=3}
\end{figure}
Figure \ref{fazk=3} (a) and (b) show the phase diagrams of the
model on Cayley tree of order three. Contrary to the Vannimenus's
work \cite{Vannimenus}, here the multicritical Libschit points
appear in non-zero points. In Figure \ref{fazk=3} (a), we observe
that the phase diagram contains ferromagnetic (F), period-2,
period-3 and modulated phases. In Figure \ref{fazk=3} (b), the
phase diagram consists of ferromagnetic (F), paramagnetic (P)
(fixed point), chaotic (C) (or modulated), and antiferromagnetic +
+ - - (four cycle antiferromagnetic phase) phases. In the chaotic
phase, small regions with periodic orbits corresponding to
commensurate phases are observed. Although it is difficult, to
distinguish long period behavior from chaos, it seems that the
behavior in the intermediate region is predominantly chaotic
\cite{MTA1985a}. Note that the modulated phase can consist of
commensurate (periodic) and incommensurate (aperiodic) regions
corresponding to the so called "devilТs staircase". In order to
distinguish these phases from each other, one needs to analysis
the modulated phase regions via Lyapunov exponent and the
attractors in detail (see
\cite{UGAT2012IJMPC,MTA1985a,Inawashiro,Inawashiro-T1983}). Here,
we will not give these details.

From Proposition \ref{proposition1}, we have the following theorem
\begin{thm}\label{theorem-p} The model \eqref{Ham-Ist} (with $x \geq 0, \alpha > 0, b > 0$) has a unique
paramagnetic phase if $b<1$. Assume that If $b
> \sqrt{2}$, then  the model \eqref{Ham-Ist} has exactly two paramagnetic
phases if either $\eta_1(b)=\alpha^{-2}$ or
$\eta_2(b)=\alpha^{-2}$. If $\eta_1(b)<\alpha^{-2}<\eta_2(b)$ and
then there exists $\eta_1(b)$, $\eta_2(b)$ with
$0<\eta_1(b)<\eta_2(b)$ such that the model \eqref{Ham-Ist} has
exactly three paramagnetic phases.
\end{thm}

\subsection{The fixed points of the operator $F$ for $b=1$}
In the equation \eqref{rec-eq3}, if we assume as $b=1$, then we
have
\begin{eqnarray}\label{rec-eqb=1}
\left\{
\begin{array}{l}
 v_1=\frac{\alpha}{c}\left(c^2v_1^{3}+v_4^{3}\right) \\
 v_4=\frac{1}{\alpha c}\left(v_5^{3}+c^2v_8^{3}\right) \\
 v_5=\frac{1}{\alpha c}\left(v_1^{3}+c^2v_4^{3}\right)\\
 v_8=\frac{\alpha}{c}\left(c^2v_5^{3}+v_8^{3}\right).
\end{array}
\right.
\end{eqnarray}

Now, we describe the positive fixed points of system
\eqref{rec-eqb=1}.
%
Let us consider the following set
\begin{eqnarray}\label{non-para}
B:=\{(v_1,v_4,v_5,v_8)\in \mathbf{R}^4_+:v_1=v_8,v_4=v_5\}.
\end{eqnarray}
\begin{rem}
From the system \eqref{rec-eqb=1}, it is clear that the equations
$v_1=v_8$ and $v_4=v_5$ don't satisfy, i.e., the set $B$ given in
\eqref{non-para} is empty. Therefore, the paramagnetic phase
regions in the phase diagrams associated with the model disappear
in $[-2,2]\times [-2,2]\subset \mathbf{R}^{2}$.
\end{rem}
In \cite{NHSS1}, we studied the phase diagram and extreme Gibbs
measures of the Ising model on a two order Cayley tree in the
presence of competing binary and ternary interactions, we have
observed that the \textbf{P} regions completely disappears in the
phase diagram associated with the model \eqref{Ham-Ist} for $b=1$.

\section{Conclusions}\label{Conclusions}
Written for both mathematics and physics audience, this paper has
a fourfold purpose: (1) to study analytically the recurrence
equations associated with the model \eqref{Ham-Ist}; (2) to obtain
numerically the paramagnetic, the ferromagnetic and period 2
regions corresponding to the sets $M_1,M_2, B$, respectively; (3)
to illustrate the fixed points of corresponding operator; (4) to
compare the numerical results to exact solutions of the model.

We state some unsolved problems that turned out to be rather
complicated and require further consideration:
\begin{enumerate}
    \item Do any other invariant sets of the operator $F$ exist?
    \item Do positive fixed points of the operator $F$ exist outside the invariant sets?
    \item Does there exist a periodic points ($p>2$) of rather cumbersome high-order equations that can be solved by
      analytic methods?
\end{enumerate}
In the first case we have already obtained the fixed points of the
operator $F$ such that $u_1=u_8$ and $u_4=u_5$. For the periodic
case, however, it is not possible to obtain all solutions
satisfying all requirements of Equations \eqref{rec-eq2} such that
$\{\textbf{u}=(u_1,u_4,u_5,u_8)\in
\mathbf{R}^4_+:F^p(\textbf{u})=\textbf{u},p>1\}$ is invariant. The
proof of this statement is involved with a number mathematical
complexity. Also, in the second case ($b=1$), to find analytically
the fixed points of the operator  $F$ is much more difficult.

By using the standard approach, we have proved the existence of
phase transition for paramagnetic phase when $J_p> 0$ and for
phase with period 2 when $J_p<0$. These results fully consistent
with numerical results in \cite{Vannimenus}. In \cite{RAU}, the
authors have analytically studied the recurrence equations and
obtain some exact results: critical temperatures and curves,
number of phases, partition function for the Ising model on a
second-order Cayley tree. The problem can numerically be examined
by the approach in \cite{Vannimenus}. In the present paper, we
analytically investigate the fixed points of the dynamical system
associated with the Ising model on a rooted Cayley tree of order
three by solving a system of nonlinear functional equations (see
\cite{GATTMJ,AT1} for details).


\end{document}